\documentclass[draftclsnofoot]{IEEEtran}
 \onecolumn
\usepackage{latexsym}
\usepackage{cite}
\usepackage{color}
\usepackage{bm}
\usepackage{amsmath, amssymb}
\usepackage[dvips]{graphics}
\usepackage{graphicx}
\usepackage{psfrag}
 \allowdisplaybreaks

\newtheorem{definition}{Definition}

\newtheorem{theorem}{Theorem}
\newtheorem{lemma}[theorem]{Lemma}

\begin{document}

\newcommand{\be}{\begin{equation}}
\newcommand{\ee}{\end{equation}}
\newcommand{\bea}{\begin{eqnarray}}
\newcommand{\eea}{\end{eqnarray}}
\newcommand{\beaa}{\begin{eqnarray*}}
\newcommand{\eeaa}{\end{eqnarray*}}

\title{Cascade and Triangular Source Coding with Side Information at the First Two Nodes}
\author{Haim Permuter and Tsachy Weissman \\
\thanks{
Author's emails: haimp@bgu.ac.il, tsachy@stanford.edu
}%
}%
%

\maketitle \vspace{-1.4cm}

\begin{abstract}%
We consider the cascade and triangular rate-distortion problem where
side information is known to the source encoder and to the first
user but not  to the second user. We characterize the
rate-distortion region for these problems. For the quadratic
Gaussian case, we show that it is sufficient to consider  jointly
Gaussian distributions, a fact that leads to an explicit solution.
\end{abstract}

\vspace{-0.0cm}
\section{Introduction}
Yamamoto \cite{Yamamoto_source_coding81} considered the cascade
source coding problem, where a source sends a message to User 1, and
then User 1 sends a message to User 2. In this paper, we extend
Yamamoto's cascade source coding problem to the case where side
information is known to the source and to User 1, but not
to User 2. The problem is depicted in Fig. \ref{f_cascade}.

 \begin{figure}[h!]{
\psfrag{b1}[][][1]{$X$} \psfrag{box1}[][][1]{Encoder}
\psfrag{box3}[][][1]{User 2}
 \psfrag{a2}[][][1]{$R_1$}
 \psfrag{t1}[][][1]{$$}

\psfrag{A1}[][][1]{$$} \psfrag{A2}[][][1]{$$}

 \psfrag{box2}[][][1]{User 1}
\psfrag{b3}[][][1]{$R_2$} \psfrag{a3}[][][1]{}
\psfrag{Y}[][][1]{$Y$} \psfrag{t2}[][][1]{$$}
\psfrag{X1}[][][1]{$\hat X_1$} \psfrag{X2}[][][1]{$\hat X_2$}
\centerline{\includegraphics[width=13cm]{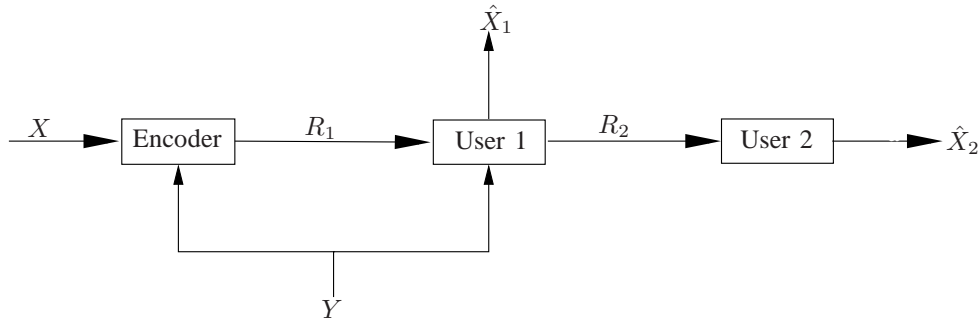}}
\caption{A cascade rate distortion problem with three nodes
(encoder, User 1, User 2), where the first two nodes have side
information $Y$. User 1 and User 2 need to reconstruct the sourse
$X$, within distortion criteria.} \label{f_cascade} }\end{figure}

More recently, Vasudevan, Tian and Diggavi
\cite{vasudevan_tian_diggavi_cascade_source06} considered the
cascade source coding problem, where side information, $Y$, is known
to the source encoder and to User 1, additional side information
$Z$ is known to User 2, and  the Markov chain $X-Z-Y$ holds.
Vasudevan et al.\cite{vasudevan_tian_diggavi_cascade_source06}
provided an inner and an outer bound and  showed that the bounds
coincide for the Gaussian case. Cuff, Su and El-Gammal
\cite{Cuff09_Gamal_Su_cascade} considered the cascade problem where
the side information is known only to the intermediate node  and
provided an inner and an outer bound. An additional related problem,
which was considered and solved in
\cite{Permuter_steinber_weissman08_volos2009},  is that of cascade
source coding when side information is known to all nodes with a
limited rate. Table \ref{t_literture} summarizes the literature  on
cascade source coding with side information.

 \begin{figure}[h!]{
\psfrag{b1}[][][1]{$X$} \psfrag{box1}[][][1]{Encoder}
\psfrag{box3}[][][1]{User 2}
 \psfrag{a2}[][][1]{$R_1$}
 \psfrag{t1}[][][1]{$$}

\psfrag{A1}[][][1]{$$} \psfrag{A2}[][][1]{$$} \psfrag{C}[][][1]{a}
\psfrag{B}[][][1]{b}\psfrag{A}[][][1]{c}

 \psfrag{box2}[][][1]{User 1}
\psfrag{b3}[][][1]{$R_2$} \psfrag{a3}[][][1]{}
\psfrag{Y}[][][1]{$Y$} \psfrag{t2}[][][1]{$$}
\psfrag{X1}[][][1]{$\hat X_1$} \psfrag{X2}[][][1]{$\hat X_2$}
\centerline{\includegraphics[width=13cm]{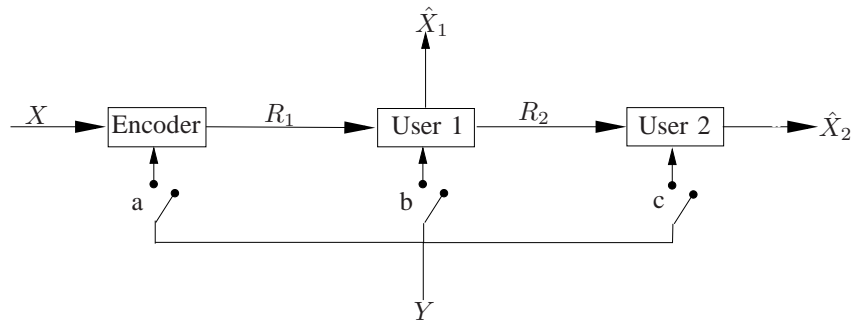}}
\caption{A cascade rate distortion problem with several options of
side information. Table \ref{t_literture} summarizes the
lietrtaure on this problem.} \label{f_all_side}
}\end{figure}

\begin{table}\label{t_literture}
\caption{literature overview of cascade source coding with side
information as shown in Fig. \ref{f_all_side}}
\begin{center}
\begin{tabular}{||c|c|c|c|c||}
\hline \hline Switch a& Switch b& Switch c& Gaussian quadratic case& General case\\
\hline \hline open & open & open& Solved \cite{Yamamoto_source_coding81}& Solved \cite{Yamamoto_source_coding81}\\
\hline open & open & closed& Solved \cite{vasudevan_tian_diggavi_cascade_source06}& Upper and lower bounds \cite{vasudevan_tian_diggavi_cascade_source06}\\
\hline open & closed & open& Upper and lower bounds \cite{Cuff09_Gamal_Su_cascade}&  Upper and lower bounds \cite{Cuff09_Gamal_Su_cascade}\\
\hline open & closed & closed& Solved \cite{vasudevan_tian_diggavi_cascade_source06}&  Upper and lower bounds \cite{vasudevan_tian_diggavi_cascade_source06}\\
\hline closed & open & open& Solved \cite{Yamamoto_source_coding81}& Solved \cite{Yamamoto_source_coding81}\\
\hline closed & open & closed& Solved \cite{vasudevan_tian_diggavi_cascade_source06}& Upper and lower bounds \cite{vasudevan_tian_diggavi_cascade_source06} \\
\hline closed & closed & open& Section \ref{s_gaussian}& Section \ref{s_problem_def}\\
\hline\hline
\end{tabular}
\end{center}
\end{table}

Of special interest in lossy source coding is the Gaussian case with
quadratic distortion, which in many source coding problems results
in an analytical solution such as in the Wyner-Ziv
problem\cite{Wyner78_gaussian_WZ} where side information is
available to the decoder, the Heegard-Berger
problem\cite{Heegard_Berger85_side_may_be_absent} where side
information at the decoder may be absent, Kaspi's problem
\cite{Kaspi94_may_be_absent,Perron_Diggavi_Emre_05_Gaussian_of_Kaspi}
where side information is known to the encoder and may or may not be
known to the decoder, the multiple description problem
\cite{GamalCover82,Ozarow_Gaussian_multiple_description}, the two-way
source coding problem \cite{Kaspi85_two_way}, the multi-terminal
problem \cite{Oohama97GaussianMultiterminal}
\cite{Wagner_Gaussian_multi_terminal}, the CEO problem
\cite{Berger_ceo_problem96,Viswanathan_Berger_ceo_Gaussian97,Oohama98_CEO_Gaussian},
 rate distortion with a helper
\cite{Vasudevan07_helper,Permuter_steinber_weissman08_helperunpublished},
 and successive refinement \cite{Equitz91Cover} and its extension to
successive refinement for the Wyner-Ziv problem
\cite{Steinberg_merhav04_sucessuve_refienment_wyner_ziv}.

 Our main result in
this paper is that the achievable region for the problem depicted in
Fig. \ref{f_cascade} is given by $\mathcal R(D_1,D_2)$, which is
defined as the set of all rate-pairs $(R_1,R_2)$ that satisfy
\begin{eqnarray}
R_2&\geq& I(Y,X;\hat X_2), \label{e_R1}  \\
R_1&\geq& I(X;\hat X_1,\hat X_2|Y),\label{e_R2}
\end{eqnarray}
for some joint distribution $P(x,y) P(\hat x_1,\hat x_2|x,y)$ for
which
\begin{eqnarray}\label{e_def_dist}
\mathbb{E}d_i(X,\hat X_i)&\leq& D_i,\ \ i=1,2.
\end{eqnarray}

An extension of the cascade source coding problem is the triangular
setting\cite{Yamamoto96TriangularSourceCoding}, where there is an additional
direct link  from the source encoder to User 2. We solve this
problem where side information exists at the source encoder and User
1, but not at User 2. 
%
%

The remainder of the paper is organized as follows. In Section
\ref{s_problem_def}, we formally define the cascade problem and
present the theorem establishing the achievable region. In Section
\ref{s_proof}, we provide a converse and achievability proofs of the
theorem, and in Section \ref{s_gaussian} we  explicitly compute the
rate region for the Gaussian case. In Section \ref{s_triangular} we
extend our result to the triangular case (cf. Fig. \ref{f_triangle}), and in Section \ref{s_extensions} we
further extend the results to multiple users and discuss the corresponding
empirical coordination problem.

\section{Cascade rate distortion: Problem definitions and main results\label{s_problem_def}}
\label{s_definition} Here we formally define the cascade
rate-distortion
 problem where side information is known to the source encoder and to User 1. We present a single-letter
characterization of the achievable region. We use the regular
definitions of rate distortion, and we follow the notation of
\cite{CovThom06}. The source sequences $\{X_i\in \mathcal X, \;
i=1,2,\cdots\}$,  and the side information sequence $\{Y_i\in
\mathcal Y,\; i=1,2,\cdots\}$ are discrete random variables drawn
from finite alphabets $\mathcal X$ and $\mathcal Y$, respectively.
The random variables $(X_i,Y_i)$ are  i.i.d. $\sim P(x,y)$. Let
$\hat{\cal X}_1$ and $\hat{\cal X}_2$ be the reconstruction
alphabets, and $d_i:\ {\cal X}\times{\hat{\cal X}_i}\rightarrow
[0,\infty)$, $i=1,2$, are single letter distortion measures.
Distortion between sequences is defined in the usual way
\begin{eqnarray}
d_i(x^n,\hat{x}_i^n) &=& \frac{1}{n} \sum_{j=1}^n
d_i (x_j,\hat{x}_{i,j}), \ \ i=1,2.
\end{eqnarray}
Let $\mathcal M_i$ denote a set of positive integers
$\{1,2,..,M_i\}$ for $i=1,2$.
\begin{definition}[Cascade rate distortion code with side information at the first two nodes]\label{def_code}
An $(n,M_1,M_2,D_1,D_2)$ code for source $X$ and side information
$Y$ consists of two encoders
\begin{eqnarray}
f_1 &:& \mathcal X^n \times \mathcal Y^n \to \mathcal M_1 \nonumber \\
f_2 &:& \mathcal Y^n \times \mathcal M_1  \to \mathcal M_2
\end{eqnarray}
and two decoders
\begin{eqnarray}
g_1 &:&  \mathcal Y^n \times \mathcal M_1
\to \hat{\cal X}_1^n \nonumber \\
g_2 &:&   \mathcal M_2 \to \hat{\cal X}_2^n
\end{eqnarray}
such that
\begin{eqnarray}\label{e_dist_cond}
\mathbb{E}\left[\frac{1}{n}\sum_{i=1}^n d_j(X_i,\hat
X_{j,i})\right]&\leq& D_j,\ \ j=1,2
\end{eqnarray}
\end{definition}
The rate pair $(R_1,R_2)$ of the $(n,M_1,M_2,D_1,D_2)$ code is
defined by
\begin{eqnarray}
R_i&=&\frac{1}{n}\log M_i; \;\; i=1,2.
\end{eqnarray}

\begin{definition}\label{def_achievable rates}
Given a distortion pair $(D_1,D_2)$, a rate pair $(R_1,R_2)$ is said
to be {\it achievable} if, for any $\epsilon>0$, and sufficiently
large $n$, there exists an
$(n,2^{nR_1},2^{nR_2},D_1+\epsilon,D_2+\epsilon)$ code for the
source $X$ with side information $Y$.
\end{definition}
\begin{definition}\label{def_the achievable_region}
{\it The (operational) achievable region} $\mathcal R^O(D_1,D_2)$ of
cascade rate distortion is
the closure of the set of all achievable rate pairs.
\end{definition}
Theorem \ref{t_cascade} is the main result of this work.
\begin{theorem}\label{t_cascade}
For the cascade rate distortion problem with side information at the
source and User 1, as depicted in Fig. \ref{f_cascade}, the
achievable region is given by
\begin{equation}
\mathcal R^O(D_1,D_2)=\mathcal R(D_1,D_2),
\end{equation}
where the region $\mathcal R(D_1,D_2)$ is defined in
(\ref{e_R1})-(\ref{e_def_dist}).
\end{theorem}

\section{Proof of Theorem \ref{t_cascade} \label{s_proof}}

{\bf Achievability:} The proof follows classical arguments, and
therefore the technical details will be omitted. We describe only
the coding structure and  justify why the indicated region is achievable. We fix
a joint distribution $P_{X,Y,\hat X_1,\hat X_2}$ for which
(\ref{e_def_dist}) holds, and an $\epsilon > 0$, and we show that
there exists a code with rates
\begin{eqnarray}
R_2&=& I(Y,X;\hat X_2)+\epsilon,   \\
R_1& = & I(X;\hat X_1,\hat X_2|Y)+3\epsilon,
\end{eqnarray}
complying with the distortion constraints.

Generate randomly $2^{n(I(X,Y;\hat X_2)+\epsilon)}$ codewords using
an i.i.d. $\sim P_{\hat X_2}$. Then bin the codewords into
$2^{n(I(X;\hat X_2|Y)+2\epsilon)}$ bins. In each bin, there are
$2^{n(I(X,Y;\hat X_2)-I(X;\hat X_2|Y)-\epsilon)}=2^{n(I(Y;\hat
X_2)-\epsilon)}$ codewords. In addition, for any typical sequences
$y^n,\hat x^n_2$ generate $2^{n(I(X;\hat X_1| Y,\hat
X_2)+\epsilon)}$ codewords using the pmf  $P(\hat x_1^n|y^n,\hat
x^n_2)=\prod_{i=1}^n P_{\hat X_1|Y,\hat X_2 }(\hat x_{1,i}|y_i,\hat
x_{2,i})$.

The source-encoder receives the sequences $x^n,y^n$ and first looks
for a codeword $\hat x_2^n$  that is jointly typical with $x^n,y^n$.
If there is such a codeword, the source encoder sends the index of the bin that includes this
codeword to User 1. User 1 looks which codeword in the received bin
is jointly typical with the side information $y^n$. Since there are
less than $2^{n(I(Y;\hat X_2)}$ in the bin, with high probability
only one codeword will be jointly typical with $y^n$ and it would be
the codeword sent by the encoder. User 1 then forwards the codeword
to User 2.

Now we can think of a new problem where the source-encoder and User
1 have side information $Y^n,\hat X_2^n$ and hence a rate $I(X;\hat
X_1|Y,\hat X_2)+\epsilon$ is needed to generate $\hat X^n_1$ that is
jointly typical with $(X^n,Y^n,\hat X_2)$. Therefore, a total rate to
User 1 of $R_1=I(X;\hat X_2|Y)+2\epsilon+I(X;\hat X_1|Y,\hat
X_2)+\epsilon=I(X;\hat X_1,\hat X_2|Y)+3\epsilon$ is needed, and an
additional rate $R_2= I(Y,X;\hat X_2)+\epsilon$ is needed from User
1 to User 2.

{\bf Converse:} Assume that we have an
$(n,M_1=2^{nR_1},M_2=2^{nR_2},D_1,D_2)$ code as in Definition
\ref{def_code}. We will show the existence of a joint distribution
$P_{X,Y,\hat X_1,\hat X_2}$ that
satisfies~(\ref{e_R1})-(\ref{e_def_dist}). Denote
$T_1=f_1(X^n,Y^n)\in\{1,...,2^{nR_1}\}$, and
$T_2=f_2(T_1,Y^n)\in\{1,...,2^{nR_2}\}$. Then,
\begin{eqnarray}\label{e_conv1}
nR_2 &\stackrel{}{\geq}&H(T_2) \nonumber \\
&\stackrel{}{\geq}&I(X^n,Y^n;T_2) \nonumber \\
&\stackrel{}{=}&\sum_{i=1}^n
H(X_i,Y_i)-H(X_i,Y_i|T_2,X^{i-1},Y^{i-1})\nonumber \\
&\stackrel{(a)}{=}&\sum_{i=1}^n
H(X_i,Y_i)-H(X_i,Y_i|\hat X_{2,i},T_2,X^{i-1},Y^{i-1})\nonumber \\
&\stackrel{}{\geq}&\sum_{i=1}^n I(X,Y;\hat X_{2,i}),
\end{eqnarray}
where equality (a) follows from the fact that the reconstruction at
time $i$,  $\hat X_{2,i}$, is a deterministic function of $T_2$. Now
consider
\begin{eqnarray}\label{e_conv2}
nR_1 &\stackrel{}{\geq}&H(T_1) \nonumber \\
&\stackrel{}{\geq}&H(T_1|Y^n) \nonumber \\
&\stackrel{(a)}{=}&H(T_1,T_2|Y^n) \nonumber \\
&\stackrel{}{\geq}&I(X^n;T_1,T_2|Y^n) \nonumber \\
&\stackrel{}{=}&\sum_{i=1}^n H(X_i|Y_i)- H(X_i|Y^n,T_1,T_2,X^{i-1})\nonumber \\
&\stackrel{(b)}{=}&\sum_{i=1}^n H(X_i|Y_i)- H(X_i|Y^n,T_1,T_2,X^{i-1},\hat X_{1,i},\hat X_{2,i})\nonumber \\
&\stackrel{}{\geq}&\sum_{i=1}^n H(X_i|Y_i)- H(X_i|Y_i,\hat X_{1,i},\hat X_{2,i})\nonumber \\
&\stackrel{}{=}&\sum_{i=1}^n I(X_i;\hat X_{1,i},\hat X_{2,i}|Y_i),
\end{eqnarray}
where equality (a) follows from the fact that $T_2$ is a
deterministic function of $T_1$ and $Y^n$, and, similarly, equality
(b) follows from the fact that $\hat X_{1,i}$ and  $\hat X_{2,i}$
are deterministic functions of $(T_1,Y^n)$ and $T_2$, respectively.

The proof is concluded in the standard way by letting $Q$ be a
random variable independent of $X^n,Y^n$,  uniformly distributed
over the set $\{1,2,3,..,n\}$, and considering the joint
distribution of $X_Q, Y_Q, \hat X_{1,Q}, \hat X_{2,Q}$. For this
joint distribution, inequalities (\ref{e_conv1}) and (\ref{e_conv2})
imply that (\ref{e_R1}) and (\ref{e_R2}) hold, respectively, and
(\ref{e_dist_cond}) implies that (\ref{e_def_dist}) holds. \hfill
\QED

\section{Cascade rate distortion: the Gaussian case \label{s_gaussian}}
In this section we explicitly calculate the rate region $\mathcal
R(D_1,D_2)$ for the cases where $X$ and $Y$ are jointly Gaussian and
the distortion is the square-error distortion. The converse and the
achievability in the previous sections are proved for the finite
alphabet case, but it can be extended to the  Gaussian case {\cite{Wyner78_gaussian_WZ}}.

Our first step in finding the achievable region for the quadratic
Gaussian case is to show that it suffices to consider only jointly
Gaussian distributions $P_{X,Y,\hat X_1, \hat X_2}$ in order to
exhaust the rate region. Then we solve an  optimization problem to
find the achievable rate-region explicitly.

\begin{lemma}[Optimality of jointly Gaussian distributions]\label{l_gauusian}
For the quadratic Gaussian cascade rate-distortion problem with side
information known to the source-encoder and to User 1, i.e., $X,Y$
are jointly Gaussian and $d_1(x,\hat x_1)=(x-\hat x_1)^2$, $d_2(x,\hat
x_2)=(x-\hat x_2)^2$, it suffices to consider only jointly Gaussian
distributions $P_{X,Y,\hat X_1, \hat X_2}$ in order to exhaust the
rate region $\mathcal R(D_1,D_2)$ given in
(\ref{e_R1})-(\ref{e_def_dist}).
\end{lemma}
\begin{proof}
Let us fix a point $(R_1,R_2,D_1,D_2)$ in the rate region and let
$P_{X,Y,\hat X_1, \hat X_2}$ be a joint distribution that satisfies
(\ref{e_R1})-(\ref{e_def_dist}). Such a distribution must exist
since  Inequalities (\ref{e_R1})-(\ref{e_def_dist}) define the rate
region (Theorem \ref{t_cascade}). Let $K$ denote the covariance
matrix induced by $P_{X,Y,\hat X_1, \hat X_2}$ and let $\tilde
P_{X,Y,\hat X_1, \hat X_2}$ denote a normal joint distribution with
mean zero and covariance matrix $K$. Now let us show that
(\ref{e_R1})-(\ref{e_def_dist}) also hold where the joint
distribution is $\tilde P_{X,Y,\hat X_1, \hat X_2}$.  Inequality
(\ref{e_def_dist}) is automatically satisfied, since it depends  on the distribution of $(X,Y,\hat{X}_1,\hat{X}_2)$ only through the covariance matrix $K$. Consider,
\begin{eqnarray}\label{e_gauss}
R_1&\geq& I(X;\hat X_1,\hat X_2|Y),\nonumber \\
&=& h(X|Y)-h(X|\hat X_1,\hat X_2,Y),\nonumber \\
&\stackrel{(a)}{=}& h(X|Y)-h(X-(\alpha_1\hat X_1+\alpha_2\hat X_2+\alpha_3 Y)|\hat X_1,\hat X_2,Y),\nonumber \\
&\stackrel{(b)}{\geq}& h(X|Y)-h(X-(\alpha_1\hat X_1+\alpha_2\hat X_2+\alpha_3 Y))\nonumber \\
&\stackrel{(c)}{\geq}& h(X|Y)-h_{\tilde P}(X-(\alpha_1\hat
X_1+\alpha_2\hat X_2+\alpha_3 Y))\nonumber \\
&\stackrel{(d)}{=}&  I_{\tilde P}(X;\hat X_1,\hat X_2|Y),
\end{eqnarray}
equality (a) is true for any set of scalars
$(\alpha_1,\alpha_2,\alpha_3)$ and in particular if we choose those
that are the linear estimator of $X$ given $\hat X_1,\hat X_2, Y$.
Note that the coefficients  $(\alpha_1,\alpha_2,\alpha_3)$  and the
variance $E(X-(\alpha_1\hat X_1+\alpha_2\hat X_2+\alpha_3 Y))^2$
 are a function only of the covariance matrix $K$. Inequality (b) follows from the fact
that conditioning reduces entropy, and (c) follows from the fact
that, given a variance, the Gaussian distribution
maximizes the differential entropy. 
The term  $I_{\tilde P}(X;\hat X_1,\hat X_2|Y)$ denotes the mutual
information induced by the Gaussian distribution $\tilde P_{X,Y,\hat
X_1, \hat X_2}$, and equality (d) follows from the fact that for the
Gaussian distribution the error, i.e.,  $X-(\alpha_1\hat
X_1+\alpha_2\hat X_2+\alpha_3 Y)$, is independent of the
observations $\hat X_1,\hat X_2,Y$.

Similarly, we have
\begin{eqnarray}
R_2&\geq& I(Y,X;\hat X_2) \nonumber   \\
&=& I(Y; \hat X_2)+I(X; \hat X_2|Y) \nonumber   \\
&\stackrel{}{\geq}& I_{\tilde P}(Y; \hat X_2)+I_{\tilde P}(X; \hat
X_2|Y),
\end{eqnarray}
where the last inequality follows from the same steps as
(\ref{e_gauss}).
\end{proof}

The next theorem provides an explicit expression for the Gaussian
case. The proof is provided in Appendix \ref{s_app_gauss_cascade}
and is based on Lemma \ref{l_gauusian} and on solving an
optimization problem with quadratic constraints and a linear
objective.
\begin{theorem}[Cascade Gaussian case] \label{t_cascade_gaussian}
The rate region of the cascade source coding with side information
at the first two nodes, where the source $X$ and the side
information $Y=X+Z$  are jointly Gaussian distributed, where $X$ and $Z$ are mutually independent, and the
distortion is quadratic, is given by
\begin{equation}\label{e_gaussian_R1}
 R_1(D_1,D_2, R_2)=\frac{1}{2}\max\left(\log
\frac{\sigma_{X|Y}^2}{\sigma_{X|W,Y}^2},\log
\frac{\sigma_{X|Y}^2}{D_1},0 \right),
\end{equation}
where $\sigma_{X|W,Y}^2$ is given by the following four cases

\begin{equation}\label{e_gaussian}
 \sigma_{X|W,Y}^2(D_1,D_2, R_2)=\left\{\begin{array}{ll}
 \left(\frac{2^{2R_2}D_2-\sigma_X^2}{\sigma_Z^2\sigma_X^2\alpha^2}
+\sigma_{X|Y}^{-2}\right)^{-1}, & \text{if } D_2\leq \sigma_{X|Y}^2
\text{ and } \frac{\sigma_X^2}{D_2}\leq 2^{2R_2}\leq
\frac{\sigma_Z^2(\sigma_X^2-D_2)}{\sigma_Z^2\sigma_X^2-D_2\sigma_Z^2-D_2\sigma_X^2}\frac{\sigma_X^2}{D_2}\\
 D_2,& \text{if } D_2\leq
\sigma_{X|Y}^2 \text{ and } 2^{2R_2}\geq
\frac{\sigma_Z^2(\sigma_X^2-D_2)}{\sigma_Z^2\sigma_X^2-D_2\sigma_Z^2-D_2\sigma_X^2}\frac{\sigma_X^2}{D_2}\\
\left(\frac{2^{2R_2}D_2-\sigma_X^2}{\sigma_Z^2\sigma_X^2\alpha^2}
+\sigma_{X|Y}^{-2}\right)^{-1}, & \text{if } D_2\geq \sigma_{X|Y}^2
\text{ and } \frac{\sigma_X^2}{D_2}\leq 2^{2R_2}\leq
\frac{\sigma_X^4}{\sigma_X^2D_2+\sigma_Z^2D_2-\sigma_X^2\sigma_Z^2}\\
\sigma_{X|Y}^2, &\text{if } D_2\geq \sigma_{X|Y}^2, \text{ and }
2^{2R_2}\geq
\frac{\sigma_X^4}{\sigma_X^2D_2+\sigma_Z^2D_2-\sigma_X^2\sigma_Z^2}\\
 \end{array}\right.
\end{equation}
and $\alpha=\left( \frac{\sigma_Z}{\sigma_X}
\sqrt{\frac{\sigma_X^2-D_2}{D_2-\sigma_X^2
2^{-2R_2}}}-1\right)^{-1}.$
\end{theorem}

Fig. \ref{f_gauss} depicts the regions for two specific values of
$D_1$ and $D_2$ such that it captures all four cases of  Eq.
(\ref{e_gaussian}).
\begin{figure}[h!]{
\psfrag{a}[][][1]{Case (a)} \psfrag{b}[][][1]{Case (b)}
\psfrag{c}[][][1]{Case (c)} \psfrag{d}[][][1]{Case (d)}
\psfrag{R1}[][][1]{$R_1$} \psfrag{R2}[][][1]{$R_2$}
\psfrag{t}[][][1]{$\downarrow$}
\psfrag{t2}[][][0.7]{$\frac{1}{2}\log\frac{\sigma_X^2}{D_2}$}
\psfrag{t5}[][][0.7]{$\frac{1}{2}\log\frac{\sigma_X^2}{D_2}$}
\psfrag{t6}[][][0.7]{$\;\;\;\;\;\;\;\;\;\;\;\;\;\;\;\;\;\;\;\;\;\;\;\;
\frac{1}{2}\log\frac{\sigma_X^4}{\sigma_X^2D_2+\sigma_Z^2D_2-\sigma_X^2\sigma_Z^2}$}

\psfrag{t3}[][][0.7]{$\;\;\;\;\;\;\;\;\;\;\;\;\;\;\;\;\;\;\;\;\;\;\;\;\;\;\;\;\;
\frac{1}{2}\log\frac{\sigma_Z^2(\sigma_X^2-D_2)}{\sigma_Z^2\sigma_X^2-D_2\sigma_Z^2-D_2\sigma_X^2}\frac{\sigma_X^2}{D_2}$}

\psfrag{t1}[][][0.7]{$\frac{1}{2}\max\left(\log
\frac{\sigma_{X|Y}^2}{D_2},\log
\frac{\sigma_{X|Y}^2}{D_1}\right)\longrightarrow\;\;\;\;\;\;\;\;\;\;\;\;\;\;\;\;\;\;\;\;\;\;\;\;\;\;\;\;\;\;\;\;\;$}
\psfrag{t4}[][][0.7]{$\;\;\;\;\;\;\;\;\;\;\;\;\;\;\;\;\;\;\;\;\;\;\;\;\;\;\;\;\;\;\;\longleftarrow
\frac{1}{2}\max\left(\log \frac{\sigma_{X|Y}^2}{D_2},\log
\frac{\sigma_{X|Y}^2}{D_1}\right)$}

\centerline{\includegraphics[width=12cm]{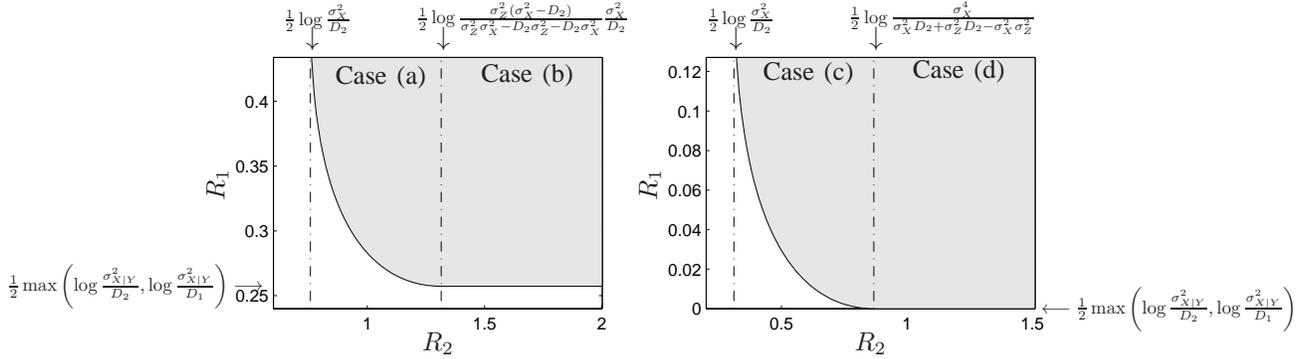}}
\caption{The Gaussian quadratic rate region. The graph on the left
hand side shows the rate region for the case where
$\sigma_X^2=\sigma_Z^2=1$, $D_2=0.35$ and $D_1=0.4$. Since
$D_2<\sigma_{X|Y}^2$, the rate region is given by Cases (a) and (b)
in Eq. (\ref{e_gaussian}). The right hand side graph shows the rate
region for the case where $\sigma_X^2=\sigma_Z^2=1$, $D_2=0.65$ and
$D_1=0.5$. Since $D_2>\sigma_{X|Y}^2$, the rate region is given by
Cases (c) and (d) in Eq. (\ref{e_gaussian})} \label{f_gauss}
}\end{figure}

Now, let us consider several extreme cases that can be easily solved
using Theorem \ref{t_cascade_gaussian}.
\subsubsection{Side information is independent of the source $ X\perp Y$} This
means that $\sigma_{X|Y}^2=\sigma_{X}^2$ and $\sigma_{Z}^2=\infty$.
For such a case (\ref{e_gaussian}) becomes


\begin{equation}
 \sigma_{X|W,Y}^2(D_1,D_2, R_2)=\left\{\begin{array}{ll}
 \sigma_{X}^{2}, & \text{if } D_2\leq \sigma_{X}^2
\text{ and } \frac{\sigma_X^2}{D_2}\leq 2^{2R_2}\leq
\frac{\sigma_X^2}{D_2}\\
 D_2,& \text{if } D_2\leq
\sigma_{X}^2 \text{ and } 2^{2R_2}\geq
\frac{\sigma_X^2}{D_2}\\
\infty, &\text{if } D_2\geq \sigma_{X}^2, \text{ and }
2^{2R_2}\geq 0\\
 \end{array}\right.
\end{equation}
and this implies that
\begin{equation}
 R_1(D_1,D_2, R_2)=\frac{1}{2}\max\left(\log
\frac{\sigma_{X}^2}{D_2},\log \frac{\sigma_{X|Y}^2}{D_1},0 \right),
\end{equation}
recovering a result that appears in the successive refinement source
coding paper \cite{Equitz91Cover}.

\subsubsection{Side information equals the source, i.e., $X=Y$}
For this case, $\sigma_{X|Y}^2=0$; hence $R_1=0$ and $2^{2R_2}\geq
\frac{\sigma_X^2}{D_2}$,  consistent with the well known rate distortion function of the Gaussian source.

\subsubsection{$R_2\to \infty$} If $D_2\leq\sigma_{X|Y}^2$ then
\begin{equation}
 R_1(D_1,D_2, R_2)=\frac{1}{2}\max\left(\log
\frac{\sigma_{X|Y}^2}{D_2},\log \frac{\sigma_{X|Y}^2}{D_1},0
\right),
\end{equation}
and if $D_2\geq\sigma_{X|Y}^2$
\begin{equation}
 R_1(D_1,D_2, R_2)=\frac{1}{2}\max\left(\log
\frac{\sigma_{X|Y}^2}{D_1},0 \right).
\end{equation}
Note that for this case  we can assume that the side information $Y$ is
known to all three nodes; hence only $\sigma_{X|Y}^2$ is manifested in
the expression.

\subsubsection{The message that User 2 receives depends only on the side
information} In this extreme case, the rate $R_2$ and the
distortion $D_2$ are large enough so that the message that User 2
receives depends only on the side information. This case is depicted
in Fig. \ref{f_cascade_extereme}.
 \begin{figure}[h!]{
\psfrag{b1}[][][1]{$X$} \psfrag{box1}[][][1]{Encoder}
\psfrag{box3}[][][1]{User 2}
 \psfrag{a2}[][][1]{$R_1$}
 \psfrag{t1}[][][1]{$$}

\psfrag{A1}[][][1]{$$} \psfrag{A2}[][][1]{$$}

 \psfrag{box2}[][][1]{User 1}
\psfrag{b3}[][][1]{$R_2$} \psfrag{a3}[][][1]{}
\psfrag{Y}[][][1]{$Y$} \psfrag{t2}[][][1]{$$}
\psfrag{X1}[][][1]{$\hat X_1$} \psfrag{X2}[][][1]{$\hat X_2$}
\centerline{\includegraphics[width=13cm]{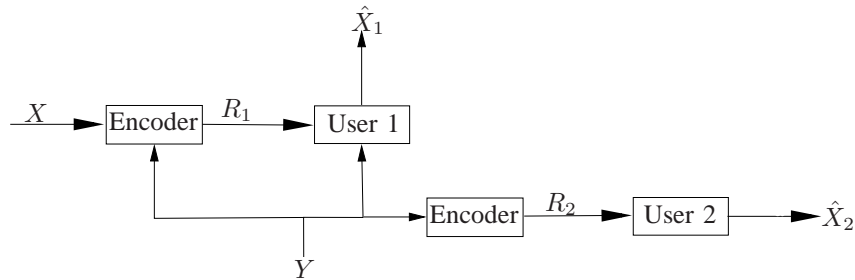}}
\caption{An extreme case where the rate $R_2$ and the distortions
$D_2$ are large enough so that the message that User 2 receives
depends only on the side information.} \label{f_cascade_extereme}
}\end{figure}

For this extreme, the rate region is simply
\begin{eqnarray}
R_1&\geq& I(X;\hat X_1|Y),\nonumber \\
R_2&\geq& I(Y;\hat X_2),
\end{eqnarray}
for all joint Gaussian distributions that satisfy $\sigma_{X|Y,\hat
X_1}^2\leq D_1$ and $\sigma_{X|\hat X_2}^2\leq D_2$.

More explicitly, this region is given by
\begin{eqnarray}
D_2&\geq& \frac{\sigma_X^2(\sigma_X^2
2^{-2R_2}+\sigma_Z^2)}{\sigma_X^2+\sigma_Z^2} \label{e_D2_side_only}\\
 R_1&\geq& \frac{1}{2}\max\left(\log
\frac{\sigma_{X|Y}^2}{D_1},0 \right).
\end{eqnarray}

Indeed, if (\ref{e_D2_side_only}) holds, then according to Theorem
\ref{t_cascade_gaussian}, $R_1(D_1,D_2,R_2)=
\frac{1}{2}\max\left(\log \frac{\sigma_{X|Y}^2}{D_1},0 \right).$

\section{Triangular source coding with side
information}\label{s_triangular}
In this section, we extend the cascade source coding discussed in
previous sections by adding a direct link from the encoder to the
second user, as depicted in  Fig. \ref{f_triangle}. The definition
of the code $(n,M_1,M_2,M_3, D_1,D_2)$ is similar to the
one given in Def. \ref{def_code} for the cascade case, with an
additional message $M_3$ at rate $R_3$  sent from the source to User 2.

 \begin{figure}[h!]{
\psfrag{b1}[][][1]{$X$} \psfrag{box1}[][][1]{Encoder}
\psfrag{box3}[][][1]{User 2}
 \psfrag{a2}[][][1]{$R_1$}
 \psfrag{t1}[][][1]{$$}

\psfrag{A1}[][][1]{$$} \psfrag{A2}[][][1]{$$}

 \psfrag{box2}[][][1]{User 1}
\psfrag{b3}[][][1]{$R_2$} \psfrag{c1}[][][1]{$R_3$}
\psfrag{a3}[][][1]{} \psfrag{Y}[][][1]{$Y$} \psfrag{t2}[][][1]{$$}
\psfrag{X1}[][][1]{$\hat X_1$} \psfrag{X2}[][][1]{$\hat X_2$}
\centerline{\includegraphics[width=10cm]{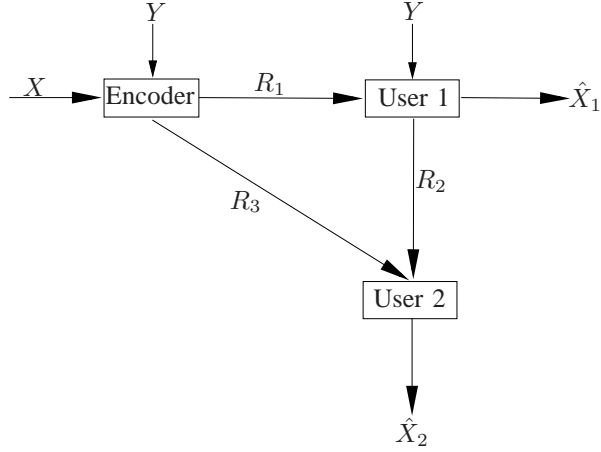}}
\caption{A triangular rate distortion problem with three nodes
(encoder, User 1, User 2), where  side information $Y$ is known to
the encoder and User 1, but not to User 2. User 1 and User 2 need to
reconstruct the sourse $X$ to within distortion criteria.}
\label{f_triangle} }
\end{figure}

\subsection{Main theorem and its proof}
\begin{theorem}[The achievable rate region for the triangular case]\label{t_triangle}
The achievable region for the problem depicted in Fig.
\ref{f_triangle} is given by $\mathcal R_{\Delta}(D_1,D_2)$, which
is defined as the set of all rate-triples $(R_1,R_2, R_3)$ that
satisfy
\begin{eqnarray}
R_1&\geq& I(X;\hat X_1,U|Y),\label{e_R1_tri}\\
R_2&\geq& I(Y,X;U), \label{e_R2_tri}  \\
R_3&\geq& I(X;\hat X_2|U),\label{e_R3_tri}
\end{eqnarray}
for some joint distribution $P(x,y) P(\hat x_1,\hat x_2, u|x,y)$ satisfying
\begin{eqnarray}\label{e_def_dist_tri}
\mathbb{E}d_i(X,\hat X_i)&\leq& D_i,\ \ i=1,2,
\end{eqnarray}
where the cardinality of the auxiliary variable  $U$ may be bounded
by $|U|\leq |\mathcal X||\mathcal Y||\mathcal {\hat X}_1||\mathcal
{\hat X}_2|+2$.
\end{theorem}

Lemma \ref{l_distribustion_tri} below shows that one can restrict the joint distribution $P(x,y) P(\hat x_1,\hat x_2, u|x,y)$ to $P(x,y) P(\hat x_1,
u|x,y)P(\hat x_2|x,u)$ without affecting the region.

{\it Proof of Converse Part of Theorem \ref{t_triangle}:} Assume that we have an
$(n,2^{nR_1},2^{nR_2},2^{nR_3},D_1,D_2)$ code. We will show the existence of a joint distribution
$P_{X,Y,U,\hat X_1,\hat X_2}$ that
satisfies~(\ref{e_R1_tri})-(\ref{e_def_dist_tri}). Denote
$T_1=f_1(X^n,Y^n)\in\{1,...,2^{nR_1}\}$, and
$T_2=f_2(T_1,Y^n)\in\{1,...,2^{nR_2}\}$, and
$T_3=f_3(X^n,Y^n)\in\{1,...,2^{nR_3}\}$. Then,
\begin{eqnarray}\label{e_conv1_tri}
nR_1 &\stackrel{}{\geq}&H(T_1) \nonumber \\
&\stackrel{}{\geq}&H(T_1|Y^n) \nonumber \\
&\stackrel{(a)}{=}&H(T_1,T_2|Y^n) \nonumber \\
&\stackrel{}{\geq}&I(X^n;T_1,T_2|Y^n) \nonumber \\
&\stackrel{}{=}&\sum_{i=1}^n H(X_i|Y_i)- H(X_i|Y^n,T_1,T_2,X^{i-1})\nonumber \\
&\stackrel{(b)}{=}&\sum_{i=1}^n H(X_i|Y_i)- H(X_i|Y^n,T_1,T_2,X^{i-1},\hat X_{1,i},U_{i})\nonumber \\
&\stackrel{}{\geq}&\sum_{i=1}^n H(X_i|Y_i)- H(X_i|Y_i,\hat X_{1,i},U_{i})\nonumber \\
&\stackrel{}{=}&\sum_{i=1}^n I(X_i;\hat X_{1,i},U_{i}|Y_i),
\end{eqnarray}
where equality (a) follows from the fact that $T_2$ is a
deterministic function of $T_1$ and $Y^n$, and, similarly, equality
(b) follows from the fact that $\hat X_{1,i}$ is a deterministic
function of $(T_1,Y^n)$ and  from defining $\hat U_{i}\triangleq
(T_2,X^{i-1},Y^{i-1})$. Now, consider
\begin{eqnarray}\label{e_conv2_tri}
nR_2 &\stackrel{}{\geq}&H(T_2) \nonumber \\
&\stackrel{}{\geq}&I(X^n,Y^n;T_2) \nonumber \\
&\stackrel{}{=}&\sum_{i=1}^n
H(X_i,Y_i)-H(X_i,Y_i|T_2,X^{i-1},Y^{i-1})\nonumber \\
&\stackrel{(a)}{=}&\sum_{i=1}^n
H(X_i,Y_i)-H(X_i,Y_i|U_i)\nonumber \\
&\stackrel{}{\geq}&\sum_{i=1}^n I(X,Y;U_i),
\end{eqnarray}
where equality (a) follows from definition of $U_i=(
T_2,X^{i-1},Y^{i-1})$. In addition, consider
\begin{eqnarray}\label{e_conv3_tri}
nR_3 &\stackrel{}{\geq}&H(T_3) \nonumber \\
&\stackrel{}{\geq}&H(T_3|T_2) \nonumber \\
&\stackrel{}{\geq}&I(X^n,Y^n;T_3|T_2) \nonumber \\
&\stackrel{}{=}&\sum_{i=1}^n
H(X_i,Y_i|T_2,X^{i-1},Y^{i-1})-H(X_i,Y_i|T_2,T_3,X^{i-1},Y^{i-1})\nonumber \\
&\stackrel{(a)}{=}&\sum_{i=1}^n
H(X_i,Y_i|U_i)-H(X_i,Y_i|\hat X_{2,i},U_i)\nonumber \\
&\stackrel{}{\geq}&\sum_{i=1}^n I(X,Y;\hat X_{2,i}|U_i)\nonumber \\
&\stackrel{}{\geq}&\sum_{i=1}^n I(X;\hat X_{2,i}|U_i),
\end{eqnarray}
where equality (a) follows from the definition of $U_i=
(T_2,X^{i-1},Y^{i-1})$ and the fact that $\hat X_{2,i}$ is a
deterministic function of $(T_2,T_3)$.

The proof is concluded in the standard way by letting $Q$ be a
random variable independent of $X^n,Y^n$,  uniformly distributed
over the set $\{1,2,3,..,n\}$, and considering the joint
distribution of $X_Q, Y_Q, U_Q, \hat X_{1,Q}, \hat X_{2,Q}$. For
this joint distribution, Inequalities (\ref{e_conv1_tri}),
(\ref{e_conv2_tri}), (\ref{e_conv3_tri})  imply that
(\ref{e_R1_tri}), (\ref{e_R2_tri}) and (\ref{e_R3_tri}) hold,
respectively, and the fact that the code we have fixed satisfies the distortion constraints implies that
(\ref{e_def_dist_tri}) holds.

To prove the cardinality bound of $U$, we invoke the support
lemma~\cite[pp. 310]{Csiszar81}.
 The external random
variable $U$ must have $|\mathcal X||\mathcal Y||\mathcal {\hat
X}_1||\mathcal {\hat X}_2|-1$ letters to preserve $P(x,y,\hat
x_1,\hat x_2)$
 plus three more to preserve the expressions
$I(X;\hat X_1,U|Y)$, $I(Y,X;U)$, $I(X;\hat X_2|U)$. Note that
preserving $P(x,y,\hat x_1,\hat x_2)$ implies that
$\mathbb{E}d_i(X,\hat X_i)\leq D_i$ for $i=1,2$ is also preserved.
 \hfill \QED

For the achievability part, we first establish the following:
%
\begin{lemma}[Optimality of $\hat X_2-(X,U)-(\hat
X_1,Y)$]\label{l_distribustion_tri}
The rate region $\mathcal R_\Delta(D_1,D_2)$, which is defined by
(\ref{e_R1_tri})-(\ref{e_def_dist_tri}), does not decrease by
restricting the joint distribution to the form $P(x,y) P(\hat x_1,
u|x,y)P(\hat x_2|x,u)$.
\end{lemma}
\begin{proof}
For a fixed $(D_1,D_2)$, let the rate-triple $(R_1,R_2,R_3)\in
\mathcal R_\Delta(D_1,D_2).$ Then there exists a joint distribution
\begin{equation}\label{e_p1}
P(x,y,u,\hat x_1,\hat x_2) =P(x,y) P(\hat x_1,\hat x_2, u|x,y),
\end{equation}
 for
which (\ref{e_R1_tri})-(\ref{e_def_dist_tri}) hold. Let $P(\hat x_1,
u|x,y)$ and $P(\hat x_2|x,u)$ be the conditional distribution
induced by $P(x,y,u,\hat x_1,\hat x_2)$. We now claim that
(\ref{e_R1_tri})-(\ref{e_def_dist_tri}) are satisfied under the
joint distribution
\begin{equation}\label{e_p2}
\tilde P(x,y,u,\hat x_1,\hat x_2) =P(x,y)P(\hat x_1, u|x,y)P(\hat
x_2|x,u).\end{equation}
 This is
true, since the expressions (\ref{e_R1_tri})-(\ref{e_def_dist_tri})
depend on $P(x,y,u,\hat x_1,\hat x_2)$ only through the marginals $P(x,y,u,\hat x_1)$ and  $P(x,u,\hat
x_2)$. Now notice that those marginals are the same whether the
joint distribution is $P(x,y,u,\hat x_1,\hat x_2)$ or $\tilde
P(x,y,u,\hat x_1,\hat x_2)$.
\end{proof}

{\it Sketch of proof of Achievability part of Theorem \ref{t_triangle}:} The achievability proof follows directly from
the achievability of cascade source coding as given in Theorem
\ref{t_cascade}. First, we fix a joint distribution of the form
$P(x,y)P(\hat x_1, u|x,y)P(\hat x_2|x,u,y)$ such that
(\ref{e_R1_tri})-(\ref{e_def_dist_tri}) hold. Since $R_1> I(X;\hat
X_1,U|Y)$ and $R_2> I(Y,X;U)$, then according to Theorem
\ref{t_cascade}, we can generate $(\hat X_1^n, U^n)$ that with high
probability would be jointly typical with $(X^n,Y^n)$ according to
the distribution $P(x,y)P(\hat x_1, u|x,y)$. Now, since $U^n$ is
known both to the encoder and to User 2, we need a rate $R_3>I(X;\hat X_2|U)$
to generate $\hat X_2^n$ such that with high
probability it is jointly typical with $X^n, U^n$. Finally,
because of the Markov relation $\hat X_2-(X,U)-(\hat X_1,Y),$ we can
invoke the Markov lemma, and conclude that the sequences $X^n,Y^n,
\hat X_1^n, , \hat X_2^n, U^n$ are jointly typical and therefore
the distortion criteria are satisfied.
\hfill \QED

%

\subsection{The Gaussian triangular case}
We now evaluate the rate region of the triangular network
depicted in Fig. \ref{f_triangle} for the quadratic Gaussian case,
i.e., $X,Y$ are jointly Gaussian and $d_1(x,\hat x_1)=(x-\hat x_1)^2$,
$d_2(x,\hat x_2)=(x-\hat x_2)^2$. We first show that it suffices to
consider only Gaussian joint distributions for exhausting the region,
and then we show that by a small change in the Gaussian cascade
region we obtain the Gaussian triangular region.
\begin{theorem}[Optimality of jointly Gaussian distributions]\label{t_gaussian_tri}
For the quadratic Gaussian triangular rate-distortion problem with
side information known to the source-encoder and to User 1, it
suffices to consider only jointly Gaussian distributions
$P_{X,Y,U,\hat X_1, \hat X_2}$ in order to exhaust the rate region
$\mathcal R_{\Delta}(D_1,D_2)$ given in
(\ref{e_R1_tri})-(\ref{e_def_dist_tri}).
\end{theorem}

Before proving the theorem, let us introduce the Pareto frontier
\cite{BoydOptimizationBook04} of a region and show that if two
rate-regions have the same Pareto frontier then they are identical.
The {\it Pareto frontier} of a region $\mathcal R$, which we denote by $
Par({\mathcal R})$, is the set of all points for which there is no
strictly better point in the region. Formally,
\begin{equation}
Par(\mathcal R)=\{R^n \in \mathcal R: \nexists \tilde R^n \in
\mathcal R \text{ s.t. } \tilde R^n \prec R^n\},
\end{equation}
where $ \tilde R^n \prec R^n $ denotes that $\tilde R_i\leq R_i$ for
all $1\leq i\leq n$ and for some $1\leq i\leq n$,  $\tilde R_i<
R_i$.

\begin{lemma}\label{l_pareto}
If two rate-regions, $\mathcal R_1$ and $\mathcal R_2$, have the same
Pareto frontier, then they are identical.
\end{lemma}
\begin{proof}
Let us show that the assumptions  $R\in \mathcal R_1$ and $R \notin
\mathcal R_2$ lead to a contradiction. If $R\in \mathcal R_1$, then
there exists a point $R_p\in Par(\mathcal R_1)$ that satisfies
$R_p\prec R$. Since $R_p\in Par(\mathcal R_1)$, it follows that
$R_p\in Par(\mathcal R_2)$. Finally, since  $R_p\in \mathcal R_2$
and $R_p\prec R$, then $R\in \mathcal R_2,$  which contradicts the
assumption.
\end{proof}

{ \it Proof of Theorem \ref{t_gaussian_tri}:} As a result of Lemma \ref{l_pareto},
we conclude that it suffices to prove  Theorem \ref{t_gaussian_tri}
only for the points in the Pareto frontier. In addition, we notice
that  points that are Pareto optimal satisfy
(\ref{e_R1_tri})-(\ref{e_R3_tri}) with equality, which may be also
written as
\begin{eqnarray}
R_1&=& I(X;\hat X_1,U|Y),\label{e_R1_tri_eq}\\
R_2&=& I(Y,X;U), \label{e_R2_tri_eq}  \\
R_3+R_2&=& I(Y, X;\hat X_2,U).\label{e_R3_tri_eq}
\end{eqnarray}
Finally, assuming without loss of generality $U$ is real-valued and
using similar arguments as in Lemma \ref{l_gauusian}, we conclude
that for any joint distribution $ P_{X,Y,\hat X_1, \hat X_2, U}$
there exists a Gaussian joint distribution, $\tilde P_{X,Y,\hat X_1,
\hat X_2, U}$, with the same covariance matrix as $ P_{X,Y,\hat X_1,
\hat X_2, U}$,  for which the induced right hand sides of
(\ref{e_R1_tri_eq})-(\ref{e_R3_tri_eq}) do not increase.\hfill \QED

Now, with a small change in the solution to the Gaussian cascade, we
obtain the triangular Gaussian region. The proof is deferred to
Appendix \ref{s_app_tri_gauss}.
 \begin{theorem}[Triangle Gaussian case] \label{t_tri_gaussian}
The rate region of the triangular source coding with side
information at the first two nodes, where the source $X$ and the
side information $Y=X+Z$  are jointly Gaussian distributed, where
$X$ and $Z$ are mutually independent, and the distortion is
quadratic, is given by Eq. (\ref{e_gaussian_R1})-(\ref{e_gaussian}),
where $D_2$ is replaced by $D_2 2^{2R_3}$ i.e.,
$R_1^{triangle}(D_1,D_2, R_2,R_3)=R_1^{cascade}(D_1,D_2 2^{2R_3},
R_2)$.
\end{theorem}
\section{Extensions \label{s_extensions}}
Here we present two further extensions. The first is obtained by
generalizing the triangular network results to  more users. The
second is obtained by considering a more general problem of
empirical coordination  rather than distortion criteria.

\subsection{Multiple Users}
\begin{figure}[h!]{
\psfrag{b1}[][][1]{$X$} \psfrag{box1}[][][0.8]{Encoder}
\psfrag{box3}[][][0.8]{User 2} \psfrag{box4}[][][0.8]{User $k$}
\psfrag{box5}[][][0.7]{User $k+1$} \psfrag{box6}[][][0.7]{User
$k+2$} \psfrag{box7}[][][0.7]{User $k+l$}

 \psfrag{a2}[][][0.9]{$R_2$}
 \psfrag{a1}[][][0.9]{$R_1$}
 \psfrag{a3}[][][0.9]{$R_3$}
 \psfrag{a4}[][][0.9]{$R_k$}
\psfrag{a5}[][][0.9]{$\;\;\;\; R_{k+1}$}
\psfrag{a6}[][][0.9]{$\;\;\;\; R_{k+2}$}
\psfrag{a8}[][][0.9]{$\;\;\;\;\;\;\; R_{k+l-1}$}
\psfrag{a7}[][][0.9]{$\;\;\;\;\; R_{k+l}$}

\psfrag{Y}[][][1]{$Y$}

\psfrag{X2}[][][1]{$\hat X_2$} \psfrag{X1}[][][1]{$\hat X_1$}
\psfrag{X6}[][][0.9]{$\hat X_k$} \psfrag{X7}[][][0.9]{$\hat
X_{k+1}$} \psfrag{X8}[][][0.9]{$\hat X_{k+2}$}
\psfrag{X9}[][][0.9]{$\hat X_{k+l}$}
\psfrag{c1}[][][0.9]{$R_{k+l+1}\;\;\;\;\;$}
 \psfrag{box2}[][][0.8]{User 1}
\centerline{\includegraphics[width=13cm]{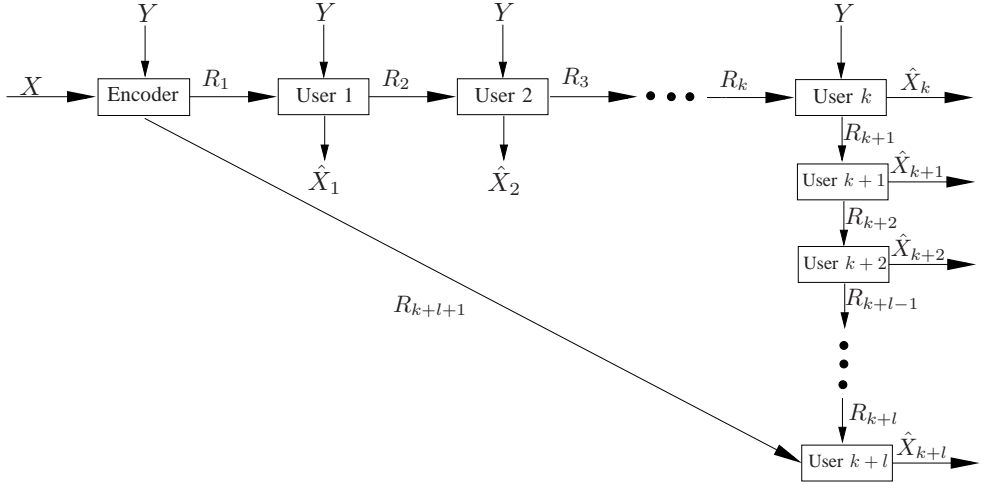}}
\caption{A triangular rate distortion problem with $k+l$ users,
where the side information $Y$ is known to the encoder and to Users
$1, 2, ..., k$, but not to Users $k+1, k+2, ..., k+l$.}
\label{f_multi} }
\end{figure}

The triangular problem depicted in Fig. \ref{f_triangle} can be
extended to $k+l$ users, where the side information is known to the
source encoder and to Users $1,2,...,k$, but is not known to Users
$k+1,k+2, ..., k+l$. This problem is depicted in Fig. \ref{f_multi},
and its region is given by the next theorem.

\begin{theorem}\label{t_multi}
The achievable region for the problem depicted in Fig.
\ref{f_multi} is given by the vector rates $(R_1,R_2, ..., R_{k+l+1})$ that
satisfy
\begin{eqnarray}
R_i&\geq& I(X;\hat X_i,\hat X_{i+1}, ..., \hat X_{k+l-1},U|Y), \ 1\leq i\leq k \nonumber \\
R_{j}&\geq& I(X;\hat X_{j}, ..., \hat X_{k+l-1},U), \ \ \ \ \ \ k+1\leq j\leq k+l \nonumber \\
R_{k+l+1}&\geq& I(X;\hat X_{k+l}|U), \nonumber \\
\end{eqnarray}
for some distribution $P(x,y) P(\hat x_1,\hat x_2,..., \hat x_k, u|x,y)$ for which
\begin{eqnarray}
\mathbb{E}d_i(X,\hat X_i)&\leq& D_i,\  \ 1\leq i\leq k+l .
\end{eqnarray}
where the cardinality of the auxiliary variable  $U$ may be bounded
by $|U|\leq |\mathcal X||\mathcal Y||\mathcal {\hat X}_1||\mathcal
{\hat X}_2|...|\mathcal {\hat X}_{k+l}|+k+l$.
\end{theorem}
The proof of Theorem \ref{t_multi} follows similar steps as the
proof of Theorem \ref{t_triangle} and is therefore omitted.

\subsection{Empirical coordination}
In \cite{CuffPermuterCover09SubmittedIT}, two coordination problems
were introduced: Empirical coordination, where the goal is to
generate sequences with a specific empirical distribution, and
strong coordination, where the goal is to generate sequences with a
distribution that is close (in total variation) to a specific i.i.d.
distribution. The empirical coordination problem is a generalization
of the rate distortion problem, since a distortion constraint
defines a half-plane in the empirical distribution space. Hence, if
we find the optimal rate needed to generate a specific empirical
distribution, we also find the optimal rate needed to generate a
specific distortion constraint.

For the cascade rate distortion problem with side information at the first two nodes, the extension to the empirical coordination problem is straightforward.
\begin{theorem}[Rate coordination in the cascade problem] The rate coordination region $R_{P_0}(P(\hat x_1,\hat x_2|x,y))$ of the cascade problem where side information is known to the first two nodes, where $X,Y\sim P_0(x,y)$, and an empirical distribution $P_0(x,y)P(\hat x_1,\hat x_2|x,y)$ is desired, is given by
\begin{eqnarray}
R_2&\geq& I(Y,X;\hat X_2), \nonumber\\
R_1&\geq& I(X;\hat X_1,\hat X_2|Y),
\end{eqnarray}
where the joint distribution evaluating the mutual information expression is $P_0(x,y)P(\hat x_1,\hat x_2|x,y)$.
\end{theorem}

\begin{proof}
The achievability proof follows immediately from the achievability
proof of Theorem \ref{t_cascade}, where we fixed  an empirical
distribution and showed that it can be achieved using the above
rates. The converse also follows from the converse of Theorem
\ref{t_cascade}, but in the last step we need to invoke
\cite[Proposition 2]{CuffPermuterCover09SubmittedIT}, which states
that the expected empirical distribution equals the distribution of
the random variables chosen uniformly over the time sequence
$1,2,...,n$, i.e., ${\mathbb E} \left[P_{X^n,Y^n,\hat X_1^n,\hat
X_2^n}(x,y,\hat x_1,\hat x_2)\right]{=}P_{X_Q,Y_Q,\hat X_{1,Q},\hat
X_{2,Q}}(x,y,\hat x_1,\hat x_2).$
\end{proof}

 However, the triangular coordination problem is an open problem, even without side information.
The solution here is heavily based on the fact that in the
achievability proof it suffices to consider only a specific
empirical distribution (with a Markov structure), but for an
arbitrary distribution the coordination problem remains open.

\appendices
\section{Proof of Theorem \ref{t_cascade_gaussian}\label{s_app_gauss_cascade}}
Following Lemma \ref{l_gauusian} we can rewrite the rate region for
the Gaussian case as:
\begin{eqnarray}
R_2&\geq& I(Y,X;W), \label{e_R2W}  \\
R_1&\geq& I(X;V,W|Y),\label{e_R1W}
\end{eqnarray}
where the vector $(X,Y,V,W)$ is jointly Gaussian distributed and
satisfies
\begin{eqnarray}
\sigma^2_{X|W}&\leq& D_2\label{e_d2W}\\
\sigma^2_{X|W,V,Y}&\leq& D_1\label{e_d1W},
\end{eqnarray}
where $\sigma^2_{A|B}\triangleq E[(A-E[A|B])^2].$

Without loss of generality let us choose the following structure
\begin{eqnarray}\label{e_ywv}
Y&=&X+Z,\nonumber \\
W&=&X+\alpha Y+Z_2=(1+\alpha)X+\alpha Z+Z_2,\nonumber \\
V&=&X+\beta Y+\gamma Z_2+Z_1,
\end{eqnarray}
where the random variables $X,Z,Z_1,Z_2$ are jointly Gaussian and
mutually independent, with variances $\sigma_X^2,\sigma_Z^2,
\sigma_{Z_1}^2, \sigma_{Z_2}^2$, respectively,  and the coefficients
$(\alpha,\beta,\gamma)$ are real number scalars.

Equations (\ref{e_R1W})-(\ref{e_d1W}) become
\begin{eqnarray}\label{e_r2_sigma}
R_2&\geq& I(X,Y;W)\nonumber \\
&=&H(W)-H(W|X,Y)\nonumber \\
&=&\frac{1}{2}\log
\frac{(1+\alpha)^2\sigma_X^2+\alpha^2\sigma_Z^2+\sigma_{Z_2}^2}{\sigma_{Z_2}^2}
\end{eqnarray}
\begin{equation}\label{e_d2_sigma}
D_2\geq
\sigma_{X|W}^2=\frac{\sigma_X^2(\alpha^2\sigma_Z^2+\sigma_{Z_2}^2)}{(1+\alpha)^2\sigma_X^2+\alpha^2\sigma_Z^2+\sigma_{Z_2}^2}
\end{equation}
\begin{equation}\label{e_R1_sigma}
R_1=\frac{1}{2}\max\left(\log
\frac{\sigma_{X|Y}^2}{\sigma_{X|W,Y}^2},\log
\frac{\sigma_{X|Y}^2}{D_1} \right),
\end{equation}
where
$\sigma_{X|Y}^2=\frac{\sigma_X^2\sigma_Z^2}{\sigma_X^2+\sigma_Z^2}$
and
$\sigma_{X|W,Y}^{-2}=\sigma_{Z_2}^{-2}+\sigma_{X}^{-2}+\sigma_{Z}^{-2}$.

Inequalities (\ref{e_r2_sigma}) and (\ref{e_d2_sigma}) follow
directly from (\ref{e_R2W}) and (\ref{e_d2W}), respectively. Eq.
(\ref{e_R1_sigma}) follows from combining the following two equations,
(\ref{e_R1_D1ge})- (\ref{e_R1_D1le}). If $D_1\geq \sigma^2_{X|W,Y}$,
then (\ref{e_d1W}) is automatically satisfied, and then $V$ is not
needed (may be independent of anything else) and therefore
\begin{eqnarray}\label{e_R1_D1ge}
R_1&\geq& I(X;W|Y) \nonumber \\
&=& H(X|Y)- H(X|Y,W)\nonumber \\
&=& H(X|Y)- H(X|Y,W)\nonumber \\
&=&\frac{1}{2}\log \frac{\sigma_{X|Y}^2}{\sigma_{X|W,Y}^2}.
\end{eqnarray}
If $D_1\leq \sigma^2_{X|W,Y}$, then
\begin{eqnarray}\label{e_R1_D1le}
R_1&\geq& I(X;V,W|Y)\nonumber\\
&=&H(X|Y)-H(X|Y,V,W)\nonumber \\
&=& \frac{1}{2}\log \frac{\sigma_{X|Y}^2}{D_1}.
\end{eqnarray}
The last equality is due to the fact that we can choose $(\beta, \gamma,
Z_1)$ such that $\sigma^2_{X|W,V,Y}= D_1$. 

Now let us fix $D_1\geq 0$,  $D_2\geq0$, and
$R_2\geq\frac{1}{2}\log\frac{\sigma_X^2}{D_2}$, and let us find the
function $R_1(D_1,D_2,R_2)$, which defines the rate region. (The
condition on $R_2$ is due to the fact that if
$R_2<\frac{1}{2}\log\frac{\sigma_X^2}{D_2}$ the rate will not be
achievable for any $R_1$). To find $R_1$ we need to solve
the following optimization problem

\begin{align}
 \text{maximize\ \ } & \sigma_{Z_2}^2\label{e_obj_opt} \\
 \text{subject to\ \ } &
(2^{2R_2}-1)\sigma_{Z_2}^2\geq
(1+\alpha)^2\sigma_X^2+\alpha^2\sigma_Z^2 \label{e_st1}\\
&\sigma_{Z_2}^2(\sigma_X^2-D_2)\leq\alpha^2(\sigma_X^2D_2+\sigma_Z^2D_2-\sigma_X^2\sigma_Z^2)+2\alpha\sigma_X^2D_2+D_2\sigma_X^2\label{e_st2}
\end{align}

The objective (\ref{e_obj_opt}) follows from the fact that $R_1$ depends only on
$\sigma_{Z_2}^2$ and (\ref{e_st1}) and (\ref{e_st2}) follow from
(\ref{e_r2_sigma}) and (\ref{e_d2_sigma}), respectively. To solve
this optimization problem, we divide the problem into four cases,
where each case has a simple solution (each case corresponds to a
line in (\ref{e_gaussian})).

Case 1: For this case we assume that
\begin{equation}\label{e_case1}
\sigma_X^2D_2+\sigma_Z^2D_2-\sigma_X^2\sigma_Z^2<0 \Rightarrow
D_2\leq
\frac{\sigma_Z^2\sigma_X^2}{\sigma_Z^2+\sigma_X^2}=\sigma^2_{X|Y},
\end{equation}
and
\begin{eqnarray}
R_2&\geq&\frac{1}{2}\log
\frac{\sigma_Z^2(\sigma_X^2-D_2)}{\sigma_Z^2\sigma_X^2-D_2\sigma_Z^2-D_2\sigma_X^2}\frac{\sigma_X^2}{D_2}.
\end{eqnarray}

%

Because of the assumption in (\ref{e_case1}), Eq. (\ref{e_st2})
holds with equality, since otherwise $\sigma_{Z_2}^2$ can be
increased until it hits the boundary of (\ref{e_st2}).

 \begin{figure}[h!]{
\psfrag{a}[][][1]{$\alpha$}\psfrag{sz}[][][1]{$\sigma_{Z_2}^2$}
\psfrag{c1}[][][1]{{\color{red} Constraint Eq. (\ref{e_st1})}\ \ \ }
\psfrag{c2}[][][1]{\color{blue} Constraint Eq. (\ref{e_st2})\ \ \ }
%
\centerline{\includegraphics[width=7cm]{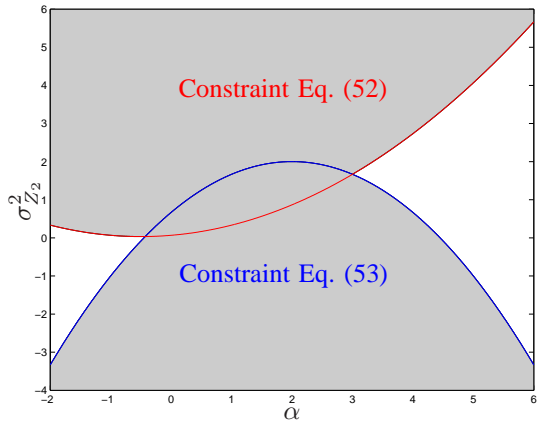}}
\caption{Case 1: the maximum of $\sigma_{Z_2}^2$, where both
constraints hold, is obtained at the maximum of Eq. (\ref{e_st2}).}
\label{f_cascad} }\end{figure}

The argument that achieves the maximum of a quadratic form
$a\alpha^2+g\alpha+c$ is $\frac{-b}{2a}$, hence the argument that
maximizes (\ref{e_st2}) is
\begin{equation} \label{e_alpha_max}
\overline
\alpha=\frac{-\sigma_X^2D_2}{\sigma_X^2D_2+\sigma_Z^2D_2-\sigma_X^2\sigma_Z^2},
\end{equation}
and the maximum is
\begin{eqnarray}
\overline \sigma^2_{Z_2}&=&c-\frac{b^2}{4a}\nonumber \\
&=&\frac{\sigma_x^2D_2}{\sigma_Z^2D_2-\sigma_Z^2\sigma_X^2}{(\sigma_X^2-D_2)(\sigma_X^2D_2+\sigma_Z^2D_2-\sigma_Z^2\sigma_X^2)}\nonumber \\
&\stackrel{}{=}&\overline \alpha\sigma_Z^2. \label{e_alphaz}
\end{eqnarray}
Note that (\ref{e_alphaz}) can be also written as
\begin{equation}\label{e_smax_simple}
\frac{1}{\overline
\sigma_{Z_2}^{2}}=\frac{1}{D_2}-\frac{1}{\sigma_Z^{2}}-\frac{1}{\sigma_X^{2}}.
\end{equation}
If $(\overline \alpha,\overline \sigma_{Z_2}^2)$ satisfy Eq.
(\ref{e_st1}), then the solution to the optimization problem is
simply $\overline \sigma_{Z_2}^2$ and using (\ref{e_R1_sigma}) we
obtain
\begin{equation}
R_1=\frac{1}{2}\max\left(\log \frac{\sigma_{X|Y}^2}{D_2},\log
\frac{\sigma_{X|Y}^2}{D_1} \right).
\end{equation}
Now let us investigate when $(\overline \alpha,\overline \sigma_{Z_2}^2)$
satisfies Eq. (\ref{e_st1}) (or equivalently (\ref{e_r2_sigma}))
\begin{eqnarray}\label{e_r2_con_case1}
R_2&\geq&\frac{1}{2}\log \frac{(1+\overline
\alpha)^2\sigma_X^2+\overline \alpha^2
\sigma_Z^2+\overline\sigma_{Z_2}^2}{\overline
\sigma_{Z_2}^2}\nonumber
\\
&\stackrel{(a)}{=}&\frac{1}{2}\log
\frac{\sigma_X^2(\overline \alpha^2\sigma_Z^2+\overline \sigma_{Z_2}^2)}{\overline \sigma_{Z_2}^2D_2}\nonumber\\
&\stackrel{(b)}{=}&\frac{1}{2}\log
\frac{\sigma_X^2(\overline \alpha^2\sigma_Z^2+\overline \alpha \sigma_{Z}^2)}{\overline \alpha \sigma_{Z}^2D_2}\nonumber\\
&\stackrel{(c)}{=}&\frac{1}{2}\log
\frac{\sigma_Z^2(\sigma_X^2-D_2)}{\sigma_Z^2\sigma_X^2-D_2\sigma_Z^2-D_2\sigma_X^2}\frac{\sigma_X^2}{D_2},\label{e_R2_cond}
\end{eqnarray}
where (a) follows from Equality (\ref{e_d2_sigma}), (b) from
(\ref{e_alphaz}) and (c) from (\ref{e_alpha_max}).

Case 2: Assume that
\begin{equation}\label{e_case1}
D_2\leq
\frac{\sigma_Z^2\sigma_X^2}{\sigma_Z^2+\sigma_X^2}=\sigma^2_{X|Y},
\end{equation}
and
\begin{eqnarray}
R_2&\leq&\frac{1}{2}\log
\frac{\sigma_Z^2(\sigma_X^2-D_2)}{\sigma_Z^2\sigma_X^2-D_2\sigma_Z^2-D_2\sigma_X^2}\frac{\sigma_X^2}{D_2}.
\end{eqnarray}
Now if (\ref{e_R2_cond}) is not satisfied, then the maximum of
$\sigma_{Z_2}^2$ should be on the boundary of the constraints,
namely, both (\ref{e_st1}) and (\ref{e_st2}) should hold with
equality. This is because the upper part of the intersection should
be either increasing or decreasing. Such a case is shown in Fig.
\ref{f_cascade2}.
\begin{figure}[h!]{
\psfrag{a}[][][1]{$\alpha$}\psfrag{sz}[][][1]{$\sigma_{Z_2}^2$}
\psfrag{c1}[][][1]{\color{red} Constraint Eq. (\ref{e_st1})\ \ \  \
\ \ \ \ \ \ \ \ } \psfrag{c2}[][][1]{\color{blue}\ \ \ \ \ \ \ \ \
Constraint Eq. (\ref{e_st2})}

\centerline{\includegraphics[width=7cm]{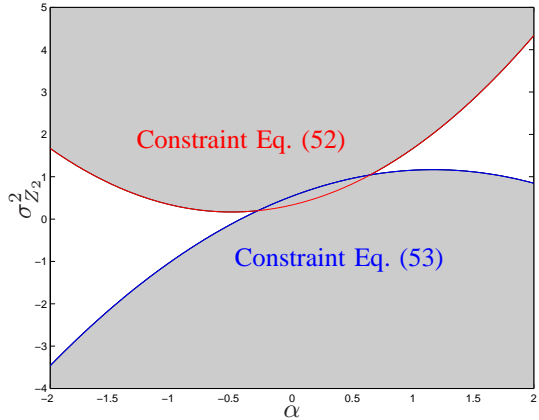}}
\caption{Case 2: the maximum of $\sigma_{Z_2}^2$, where both
constraints hold, is obtained at the intersection of (\ref{e_st1})
and (\ref{e_st2}).} \label{f_cascade2} }\end{figure}

Consider the case where (\ref{e_r2_sigma}) and (\ref{e_d2_sigma})
hold with equality. Then we obtain
\begin{equation}
2^{2R_2}\sigma_{Z_2}^2=\frac{\sigma_X^2(\alpha^2\sigma_Z^2+\sigma_{Z_2}^2)}{D_2},
\end{equation}
which implies
\begin{equation}\label{e_sigmaz2_case1b}
\sigma_{Z_2}^2=\frac{\sigma_Z^2\sigma_X^2}{2^{2R_2}D_2-\sigma_X^2}\alpha^2.
\end{equation}
Now substituting $\sigma_{Z_2}^2$ given by (\ref{e_sigmaz2_case1b})
into (\ref{e_st1}) we obtain
\begin{equation}
\frac{\alpha^2\sigma_Z^2\sigma_X^2(2^{2R_2}-1)}{2^{2R_2}D_2-\sigma_X^2}=(1+\alpha)^2\sigma_X^2+\alpha^2\sigma_Z^2,
\end{equation}
which simplifies to
\begin{equation}
\frac{\alpha^2\sigma_Z^2(\sigma_X^2-D_2)}{D_2-\sigma_X^22^{-2R_2}}=(1+\alpha)^2\sigma_X^2.
\end{equation}
Taking the square-root on each side of the equation we obtain two
possible solutions for $\alpha$:
\begin{equation}
\frac{1}{\alpha}=\pm \frac{\sigma_Z}{\sigma_X}
\sqrt{\frac{\sigma_X^2-D_2}{D_2-\sigma_X^2 2^{-2R_2}}}-1.
\end{equation}
Since we need to maximize $\sigma_{Z_2}^2$, which is proportional to
$\alpha^2$ (see Eq. (\ref{e_sigmaz2_case1b})), we choose the
solution with the plus sign.

Case 3: Assume that
\begin{equation}\label{e_case1}
D_2\geq
\frac{\sigma_Z^2\sigma_X^2}{\sigma_Z^2+\sigma_X^2}=\sigma^2_{X|Y},
\end{equation}
and
\begin{eqnarray}
R_2&\geq&\frac{1}{2}\log
\frac{\sigma_Z^2(\sigma_X^2-D_2)}{\sigma_Z^2\sigma_X^2-D_2\sigma_Z^2-D_2\sigma_X^2}\frac{\sigma_X^2}{D_2}.
\end{eqnarray}

 \begin{figure}[h!]{
\psfrag{a}[][][1]{$\alpha$}\psfrag{sz}[][][1]{$\sigma_{Z_2}^2$}
\psfrag{c1}[][][1]{{\color{red} \ \ \ \ \ \ \ \ \ \ \ Constraint Eq.
(\ref{e_st1})}} \psfrag{c2}[][][1]{\color{blue} Constraint Eq.
(\ref{e_st2})\ \ \ }
\centerline{\includegraphics[width=7cm]{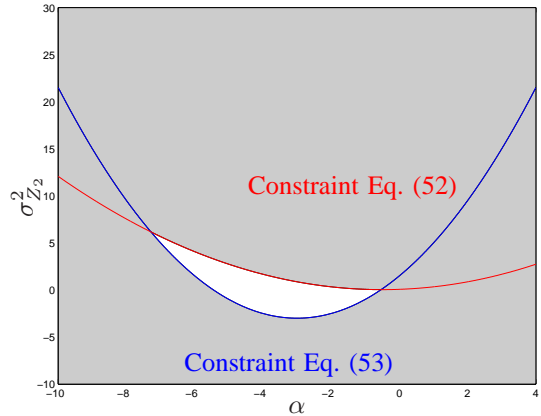}}
\caption{Case 3: the maximum of $\sigma_{Z_2}^2$, where both
constraints hold, is obtained at infinity, since there is a infinite
overlap between the constraints. \label{f_cascad3}}}\end{figure}

If
\begin{equation}
\frac{(\sigma_X^2D_2+\sigma_Z^2D_2-\sigma_X^2\sigma_Z^2)}{\sigma_X^2-D_2}\geq
\frac{\sigma_X^2+\sigma_Z^2}{2^{2R_2}-1},
\end{equation}
which is equivalent to
\begin{equation}\label{e_r2_cond_case2}
2^{2R_2}\geq
\frac{\sigma_X^4}{\sigma_X^2D_2+\sigma_Z^2D_2-\sigma_X^2\sigma_Z^2},
\end{equation}
 then the maximum of $\sigma_{Z_2}^2$ is
obtained at infinity (as illustrated in Fig. \ref{f_cascad3}), which
implies that
\begin{equation}
R_1=\frac{1}{2}\max\left(0,\log \frac{\sigma_{X|Y}^2}{D_1}
\right)=\frac{1}{2}\log \frac{\sigma_{X|Y}^2}{D_1}.
\end{equation}

 \begin{figure}[h!]{
\psfrag{a}[][][1]{$\alpha$}\psfrag{sz}[][][1]{$\sigma_{Z_2}^2$}
\psfrag{c1}[][][1]{{\color{red} Constraint Eq. (\ref{e_st1})} \ \ \
\ \ \ \ \ \ } \psfrag{c2}[][][1]{\color{blue} Constraint Eq.
(\ref{e_st2})\ \ \ \ \ \ \ \ \ \ \ \ \ \ \ \ \ \ \ \  }
\centerline{\includegraphics[width=7cm]{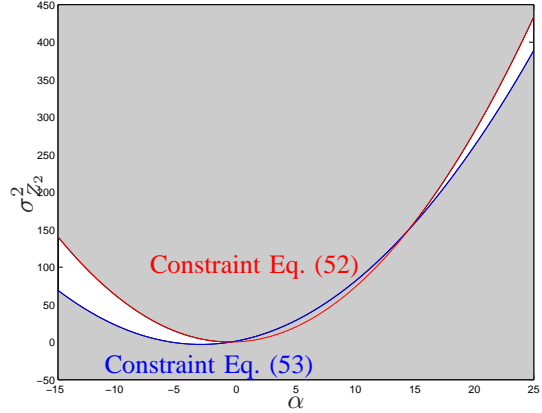}}
\caption{Case 4: the maximum of $\sigma_{Z_2}^2$, where both
constraints hold, is obtained at the intersection of (\ref{e_st1})
and (\ref{e_st2}).}\label{f_cascade4} }\end{figure}

Case 4: Assume that
\begin{equation}\label{e_case1}
D_2\geq
\frac{\sigma_Z^2\sigma_X^2}{\sigma_Z^2+\sigma_X^2}=\sigma^2_{X|Y},
\end{equation}
and
\begin{eqnarray}
R_2&\leq&\frac{1}{2}\log
\frac{\sigma_Z^2(\sigma_X^2-D_2)}{\sigma_Z^2\sigma_X^2-D_2\sigma_Z^2-D_2\sigma_X^2}\frac{\sigma_X^2}{D_2}.
\end{eqnarray}

If (\ref{e_r2_cond_case2}) does not hold, then the maximum of
$\sigma_{Z_2}^2$ should be at boundary of the constraint, namely,
(\ref{e_st1}) and (\ref{e_st2}) should hold with equality. This is
because the upper part of the intersection should be either
increasing or decreasing. Such a case is shown in Fig.
\ref{f_cascade4}. \hfill \QED

\section{Proof of Theorem \ref{t_gaussian_tri}\label{s_app_tri_gauss}}
Let us rewrite the rate region equations similarly to
(\ref{e_R1W})-(\ref{e_d1W}) as,
\begin{eqnarray}
R_1&\geq& I(X;V,W|Y),\label{e_R1W_tri}\\
R_2&\geq& I(Y,X;W), \label{e_R2W_tri}  \\
R_3&\geq& I(X;W'|W), \label{e_R3W_tri}
\end{eqnarray}
where the vector $(X,Y,V,W)$ is jointly Gaussian distributed and
satisfies
\begin{eqnarray}
\sigma^2_{X|W,W'}&\leq& D_2\label{e_d2W_tri}\\
\sigma^2_{X|W,V,Y}&\leq& D_1\label{e_d1W_tri},
\end{eqnarray}
Without loss of generality, we may assume that $X,Y,W,V$ have the
same structure as in (\ref{e_ywv}) and $W'=X+\eta W+Z'$ where
$Z'\sim N(0,\sigma_{Z'}^2$ is independent of  $X,Y,W,V$.
Furthermore, we note that we can assume that (\ref{e_R3W_tri}) holds
with equality, since if not, we can change  $\eta$ and $Z'$ such
that equality will hold, and the change will only decrease
$\sigma^2_{X|W,W'}$ -  therefore (\ref{e_R1W_tri})-(\ref{e_d1W_tri})
will continue to hold. Now, the equality in (\ref{e_R3W_tri})
implies that
\begin{equation}
 \sigma_{X|W,W'}^2=\sigma_{X|W}^2 2^{-2R_3}.
\end{equation}
Hence (\ref{e_d2W_tri}) becomes
\begin{equation}
\sigma_{X|W}^2 \leq D_2 2^{2R_3}.
\end{equation}
 Now we note that we obtain the same optimization problem as in (\ref{e_r2_sigma})-(\ref{e_R1_sigma}), just that $D_2$ is replaced by $D_2
 2^{2R_3}$.\hfill \QED

\newpage

\bibliographystyle{unsrt}
\bibliographystyle{IEEEtran}

\end{document}